\documentclass[pra,onecolumn]{revtex4-2}
\usepackage{enumitem}
\usepackage{amsfonts}
\usepackage{mathrsfs}
\usepackage{graphicx}
\usepackage{amssymb}
\usepackage{pifont}

\usepackage{amsmath}
\usepackage{hyperref}

\usepackage{epstopdf}
\usepackage{cleveref}
\usepackage{amsthm}
\usepackage{physics}
\usepackage{bm}
\usepackage{booktabs}

\hypersetup{
    colorlinks=true,
    linkcolor=blue,
    citecolor=blue,
    urlcolor=blue
}

\newtheorem{theorem}{Theorem}
\newtheorem{corollary}{Corollary}
\newtheorem{lemma}{Lemma}

\newtheorem{example}{Example}

\begin{document}
\title{Separability criteria and lower bounds of entanglement via realignment}
\author{Yuhang Xie}
\affiliation{School of Mathematics and Statistics, Henan University, Kaifeng, 475004, China}
\author{Yanjun Chu}
\email{chuyj@henu.edu.cn}
\affiliation{School of Mathematics and Statistics, Henan University, Kaifeng, 475004, China}
\author{Shao-Ming Fei}
\email{feishm@cnu.edu.cn}
\affiliation{School of Mathematics Sciences, Capital Normal University, Beijing, 100048, China}

\begin{abstract}
The detection and quantification of quantum entanglement are central problems in quantum information science. Based on the realignment of density matrices and the vectorization of reduced density matrices, we construct a family of enhanced separability criteria for bipartite and multipartite quantum systems. Our criteria recover the previously proposed realignment-based criteria as special cases and detect more entanglement than the earlier ones. Moreover, we derive tighter lower bounds for the concurrence and the convex-roof extended negativity, criteria for detecting the genuine tripartite entanglement, and lower bounds on genuine-multipartite-entanglement concurrence. The advantages of our results are demonstrated with detailed examples.
\end{abstract}

\maketitle

\section{Introduction}
Quantum entanglement is one of the most remarkable features of quantum mechanics and lies at the heart of quantum information science. As a fundamental physical resource, it plays an indispensable role in various quantum information processing tasks, such as quantum computation \cite{DivincenzoScience255,BennettNature247}, quantum cryptography \cite{EkertPRL661,FuchsPRA1163,DeutschPRL2818,HilleryPRA1829,LongPRA032302}, quantum teleportation \cite{BennettPRL1895,AlbeverioPRA012301,Gao50001}, and dense coding \cite{Bennett2881}. A fundamental problem in the theory of quantum entanglement is the detection and quantification of entanglement. Let $\mathcal{H}_A$ ($\mathcal{H}_B$) denote the Hilbert space associated with systems $A$ ($B$). A bipartite state \(\rho\in\mathcal{H}_A\otimes\mathcal{H}_B\) is said to be separable if it can be written as the convex combination of the products of local states,
\begin{equation}\label{sp}
\rho = \sum_k p_k \rho^A_k \otimes \rho^B_k,
\end{equation}
where $p_k\ge0$ with $\sum_k p_{k}=1$ and $\rho^A_k$, $\rho^B_k$ are pure states on the subsystems $A$, $B$, respectively. Otherwise, it is said to be entangled. Over the past decades, many separability criteria have been established, such as the positive partial transpose (PPT) criterion \cite{PeresPRL1413,HorodeckiPLA1996,HorodeckiPLA333}, the covariance matrix criterion \cite{GuhnePRL130504}, and the correlation matrix criterion \cite{VicenteJPA065309}. Nevertheless, the efficient and universal detection of entanglement still remains a central and largely open challenge.

Among the separability criteria, the realignment criterion \cite{Rudolph2003,Chen2003matrix} is particularly useful as it detects some bound entangled states that escape detection by the PPT criterion.  It says that the realigned matrix \(\mathcal{R}(\rho)\) of any separable state \(\rho\) satisfies \(\|\mathcal{R}(\rho)\|_{\mathrm{tr}} \leq 1\). Subsequently, Zhang \textit{et al.} \cite{Zhang2008} improved the realignment criterion by considering the realignment of \(\rho-\rho_A\otimes\rho_B\). Shen \textit{et al.} \cite{Shen2015} introduced a matrix constructed by combining the realigned matrix with the vectorization of the reduced density matrices. For multipartite systems, Zhang \textit{et al.} \cite{ZhangQIP106} proposed sequential realignment criteria, and Qi \cite{Qi2024} investigated the detection and quantification of genuine multipartite entanglement (GME) based on the realignment of density matrices. In addition to GME, the finer structure of multipartite entanglement, described by notions such as k-separability, k-nonseparability, and k-partite entanglement, has also been studied from different perspectives \cite{HongLuo2016,HongLuoSong2015,HongGaoYan2021,QiPangHou2025,GuoHouQiHe2025,QiGaoYan2017}. 

More recently, a fruitful line of work has emerged by enlarging the realigned matrix with additional blocks. Shi and Sun \cite{Shi2023family} constructed a matrix \(\mathcal{N}_{\alpha,\beta}(\rho)\) with scalar parameters \(\alpha\) and \(\beta\); Sun \textit{et al.} \cite{SunPRA012415} further generalized it to \(\mathcal{M}_{\alpha,\beta}^l(\rho)\) with an additional integer $l$; Lu \textit{et al.} \cite{LuQIP134} replaced the scalar parameters with two vectors \(\bm{a}\) and \(\bm{b}\). A different route to construct such enlarged matrices relies on symmetric measurements rather than the direct rearrangement of density matrix elements. Lu \textit{et al.} \cite{LuPRA062424} employed the \((N,M)\)-POVMs \(\{E_{\alpha,k} \mid k=1,2,\dots,M\}\ (\alpha=1,2,\dots,N)\) to arrange the probabilities of local measurement outcomes into a matrix: the lower-right block is formed by the joint probabilities \(\mathrm{Tr}[(E_{\alpha,i}\otimes E_{\beta,j})\rho]\), while the border blocks are filled with the reduced-state probabilities \(\mathrm{Tr}(E_{\alpha,i}\rho_A)\) and \(\mathrm{Tr}(E_{\beta,j}\rho_B)\). Because this construction uses local measurement probabilities in place of vectorized reduced density matrices, it can be implemented experimentally without reconstructing the full density matrix. This construction was subsequently enhanced in Ref.~\cite{LuAQT70344}, where scalar parameters were replaced by adjustable parameter vectors. Refs.~\cite{LaiLuo2023,QiPangHou2025QIP} also constructed correlation matrices via local measurements to derive separability criteria. Recently, we note that Xu \textit{et al}. \cite{XuAQTe704532026} constructed correlation matrices using the Type-II Discrete Cosine Transform to derive a class of separability criteria, which unifies and significantly generalizes the computable cross-norm and realignment criterion (CNNR) \cite{Chen2003matrix,RudolphPRA0323122003,RudolphQIC2192005,Zhang2008,Shen2015}, the de Vicente criterion \cite{VicenteQIC6242007}, as well as a large family of existing separability criteria formulated within the Bloch representation framework \cite{LiPRA0223252014,ShenSciRep288502016,ZhuQIP1852023}. In all these constructions presented in Refs.~\cite{Shi2023family,SunPRA012415,LuQIP134}, the upper-left block is either a scalar or a matrix of rank at most one. In this work, we advance this idea by permitting the upper-left block to be of arbitrary rank, thereby unifying and strengthening the aforementioned quantum separability criteria.

The quantification of entanglement constitutes another central topic in quantum information theory. Among various entanglement measures, the concurrence and the convex-roof extended negativity (CREN) are two widely studied quantities \cite{RungtaPRA042315,VidalPRA032314,LeePRA062304,ChenAlbeverioFei,deVicentePRA052320,MaGMEconcurrence,HuberPRA042328,ChenQIP3761,ZhuQIP30,LiQIP130}. However, both concurrence and CREN are defined via convex-roof extensions, which require minimization over all pure-state decompositions and are generally hard to compute. Accordingly, deriving tight, computable lower bounds for these measures is highly desirable. From the arbitrary-rank criterion, we derive lower bounds for the concurrence and CREN of bipartite states, together with a lower bound for the GME concurrence of tripartite states.

The remainder of this paper is structured as follows. In Sec.~\ref{sectII}, we define the generalized matrix \(\mathcal{Q}_{P,Q,\mu,\nu}(\rho)\) and establish the corresponding separability criterion, based on which we derive lower bounds for the concurrence and CREN of bipartite quantum states. In Sec.~\ref{sectIII}, we extend the proposed criterion to multipartite quantum systems and present criteria identifying biseparable and fully separable states. We further derive a lower bound for the GME concurrence of tripartite states. Finally, conclusions and discussions are given in Sec.~\ref{sectIV}.

\section{Detection and quantification of entanglement for bipartite systems}\label{sectII}

\subsection{Separability criteria for bipartite quantum states}

Let $\mathbb{C}^{m \times n}$ be the set of all $m \times n$ matrices over the complex field $\mathbb{C}$, and $\mathbb{R}$ be the real number field. For a matrix $A=(a_{ij}) \in \mathbb{C}^{m \times n}$, the vectorization of matrix $A$ is defined as $\operatorname{vec}(A)=(a_{11}, \cdots, a_{m1}, a_{12}, \cdots, a_{m2}, \cdots, a_{1n}, \cdots, a_{mn})^T$, where $T$ stands for the transpose.
Let $Z$ be an $m \times m$ block matrix with subblocks $Z_{i, j} \in \mathbb{C}^{n \times n}$, $i, j=1, \cdots, m$. The realigned matrix $\mathcal{R}(Z)$ of $Z$ is defined by
\begin{equation*}
	\mathcal{R}(Z)=
	\begin{pmatrix}
		\operatorname{vec}(Z_{1,1})^T \\
		\vdots \\
		\operatorname{vec}(Z_{m, 1})^T \\
		\vdots \\
		\operatorname{vec}(Z_{1, m})^T \\
		\vdots \\
		\operatorname{vec}(Z_{m, m})^T
	\end{pmatrix}.
\end{equation*}

The realignment criterion proposed in Ref. \cite{Chen2003matrix} states that any separable state $\rho$ on the bipartite space $\mathbb{C}^{m} \otimes \mathbb{C}^{n}$ obeys the inequality $\|\mathcal{R}(\rho)\|_\mathrm{tr} \leq 1$, where the trace norm of a matrix $A$ is defined as $\|A\|_\mathrm{tr}=\operatorname{Tr}\left(\sqrt{A^\dagger A}\right)$.
For a general bipartite state $\rho_{AB}$, Shi and Sun \cite{Shi2023family} introduced the following matrix:
\begin{equation}\label{realignment01}\mathcal{N}_{\alpha,\beta}(\rho)=
	\begin{pmatrix}
		\alpha\beta & \alpha\operatorname{vec}(\rho_B)^T \\
		\beta\operatorname{vec}(\rho_A) & \mathcal{R}(\rho)
	\end{pmatrix},
\end{equation}
where $\alpha,\beta\in \mathbb{R}$. They proved that for any separable state $\rho$,
$\|\mathcal{N}_{\alpha,\beta}(\rho)\|_{\mathrm{tr}} \leq \sqrt{(\alpha^2+1)(\beta^2+1)}.$
Subsequently, Sun et al. \cite{SunPRA012415} further constructed a generalized matrix $\mathcal{M}_{\alpha,\beta}^l(\rho)$ for bipartite states $\rho_{AB}$, which takes the form:
\begin{equation}\label{realignment02}
	\mathcal{M}_{\alpha, \beta}^l(\rho)=
	\begin{pmatrix}
		\alpha\beta E_{l \times l} & \alpha\omega_l(\rho_B)^T \\
		\beta\omega_l(\rho_A) & \mathcal{R}(\rho)
	\end{pmatrix},
\end{equation}
where $\alpha$ and $\beta$ are arbitrary real constants, $l\in\mathbb{N}$, $E_{l \times l}$ denotes the $l\times l$ all-ones matrix, and $\omega_l(X)$ represents the matrix consisting of $l$ identical columns with each column being $\operatorname{vec}(X)$, i.e., $\omega_l(X)=(\underbrace{\operatorname{vec}(X),\cdots,\operatorname{vec}(X)}_{l\text{ columns}})$. It is shown that any separable bipartite state $\rho\in\mathbb{C}^{m} \otimes \mathbb{C}^{n}$ satisfies
$$\left\|\mathcal{M}_{\alpha, \beta}^l(\rho)\right\|_{\mathrm{tr}} \leq \sqrt{(l\alpha^2+1)(l \beta^2+1)}.$$
Moreover, Lu \textit{et al}. \cite{LuQIP134} introduced the matrix $\mathcal{Q}_{\bm{a},\bm{b}}(\rho)$ with respect to a bipartite state $\rho_{AB}$ as
\begin{equation}\label{realignment03}
	\mathcal{Q}_{\bm{a},\bm{b}}(\rho)=
\begin{pmatrix}
	\bm{a}\bm{b}^T & \bm{a}\operatorname{vec}(\rho_B)^T \\
	\operatorname{vec}(\rho_A)\bm{b}^T & \mathcal{R}(\rho)
\end{pmatrix},
\end{equation}
where $\bm{a}=(a_1,\cdots,a_m)^T\in\mathbb{R}^{m \times 1}$ and $\bm{b}=(b_1,\cdots,b_n)^T\in\mathbb{R}^{n \times 1}$. They demonstrated that if $\rho$ is a separable state, then
\begin{equation*}
	\left\|\mathcal{Q}_{\bm{a},\bm{b}}(\rho)\right\|_{\mathrm{tr}} \leq \sqrt{\left(|\bm{a}|^2+1\right)\left(|\bm{b}|^2+1\right)}.
\end{equation*}

For a bipartite state $\rho_{AB}$, we define the following matrix,
\begin{equation}\label{realignment matrix}
	\mathcal{Q}_{P,Q,\mu,\nu}(\rho)=
	\begin{pmatrix}
		PQ^T & P\mu\operatorname{vec}(\rho_B)^T \\
		\operatorname{vec}(\rho_A)\nu^TQ^T & \mathcal{R}(\rho)
	\end{pmatrix},
\end{equation}
where $P\in\mathbb{R}^{m \times r}, Q\in\mathbb{R}^{n \times r}$, $\mu, \nu\in\mathbb{R}^{r \times 1}$ with $\nu^T\mu=1$. This matrix includes the matrices (\ref{realignment01}), (\ref{realignment02}) and (\ref{realignment03}) as particular ones.
If \(m\geq r\), \(n\geq r\), and \(\mathrm{rank}(P)=\mathrm{rank}(Q)=r\), we have \(\mathrm{rank}(PQ^T)=r\).
If $m=n=r=1$, $\mu=\nu=1$, $P=\alpha$ and $Q=\beta$, the matrix (\ref{realignment matrix}) reduces to the matrix (\ref{realignment01}) constructed in Ref.~\cite{Shi2023family}. If $r=1$, $\mu=\nu=1$, $m=n=l$, and $P=(\underbrace{\alpha,\cdots,\alpha}_{l\text{ columns}})^T$, $Q=(\underbrace{\beta,\cdots,\beta}_{l\text{ columns}})^T$, the matrix (\ref{realignment matrix}) reduces to the matrix (\ref{realignment02}) constructed in Ref.~\cite{SunPRA012415}. If $r=1$, $\mu=\nu=1$, $P=\bm{a}$, and $Q=\bm{b}$, then the matrix (\ref{realignment matrix}) reduces to the matrix (\ref{realignment03}) constructed in Ref.~\cite{LuQIP134}.
The following lemma presents two fundamental properties of the matrix $\mathcal{Q}_{P,Q,\mu,\nu}(\rho)$. The detailed proof is provided in Appendix \ref{appendixA}.

\begin{lemma}\label{lem1}
(1) \textbf{Linearity.} Let $\rho_k \in \mathcal{H}_A \otimes \mathcal{H}_B$ and $p_k \in \mathbb{R}_{\geq 0}$, $k=1,2,\dots,n$, such that $\sum\limits_{k=1}^{n}p_k=1$. Then
\begin{equation*}
	\mathcal{Q}_{P,Q,\mu,\nu}\left(\sum_{k=1}^{n}p_k\rho_k\right)=\sum_{k=1}^{n}p_k\mathcal{Q}_{P,Q,\mu,\nu}(\rho_k).
\end{equation*}
(2) \textbf{Local unitary invariance.} Let $U$ and $V$ be unitary matrices on subsystems $A$ and $B$, respectively. For any $\rho\in\mathbb{C}^{m}\otimes\mathbb{C}^{n}$, we have
\begin{equation*}
	\Big\|\mathcal{Q}_{P,Q,\mu,\nu}\big((U \otimes V)\rho(U \otimes V)^{\dagger}\big)\Big\|_{\mathrm{tr}}
	=\big\|\mathcal{Q}_{P,Q,\mu,\nu}(\rho)\big\|_{\mathrm{tr}},
\end{equation*}
where subsystems A and B have dimensions m and n, respectively.
\end{lemma}

\begin{theorem}\label{thm1}
If a bipatite state $\rho\in\mathcal{H}_A\otimes\mathcal{H}_B$ is separable, then
\begin{equation*}
	\big\|\mathcal{Q}_{P,Q,\mu,\nu}(\rho)\big\|_{\mathrm{tr}} \leq \mathfrak{M}(P,Q,\mu,\nu),
\end{equation*}
where
\begin{equation}\label{threshold} \mathfrak{M}(P,Q,\mu,\nu):=\Big\|\Big(P^TP+\nu\nu^T\Big)^{1/2}\Big(Q^TQ+\mu\mu^T\Big)^{1/2}\Big\|_{\mathrm{tr}}.
\end{equation}
\end{theorem}

\begin{proof}
If $\rho$ is separable, it can be written as (\ref{sp}), $\rho=\sum_k p_k \rho^A_k\otimes\rho^B_k$. For any product state $\rho^A_k \otimes \rho^B_k$,
\begin{equation*}
	\mathcal{Q}_{P,Q,\mu,\nu}\left(\rho^A_k \otimes \rho^B_k\right)=U_k V_k^T,~~~
	U_k=\begin{pmatrix}P \\ \operatorname{vec}(\rho^A_{k})\nu^T\end{pmatrix},~~~
	V_k=\begin{pmatrix}Q\\ \operatorname{vec}(\rho^B_k)\mu^T\end{pmatrix}.
\end{equation*}
Since $\rho^A_k,\rho^B_k$ are pure states, $\|\operatorname{vec}(\rho^A_k)\|^2=\operatorname{Tr}[(\rho^A_{k})^2]=1$ and $\|\operatorname{vec}(\rho^B_k)\|^2=\operatorname{Tr}[(\rho^B_k)^2]=1$. Hence
\begin{equation*}
	U_k^TU_k = P^TP + \nu\nu^T,~~~V_{k}^TV_k = Q^TQ + \mu\mu^T.
\end{equation*}

Using the identity $\|XY^T\|_{\mathrm{tr}}=\left\|\left(X^TX\right)^{1/2}\left(Y^TY\right)^{1/2}\right\|_{\mathrm{tr}}$, we get
\begin{equation*}
	\left\|\mathcal{Q}_{P,Q,\mu,\nu}(\rho^A_k\otimes\rho^B_k)\right\|_{\mathrm{tr}}
	=\left\|\left(U_k^TU_k\right)^{1/2}\left(V_k^TV_k\right)^{1/2}\right\|_{\mathrm{tr}}
	=\mathfrak{M}(P,Q,\mu,\nu).
\end{equation*}
Therefore, by Lemma \ref{lem1} (i) and the triangle inequality, we have
\begin{equation*}
	\big\|\mathcal{Q}_{P,Q,\mu,\nu}(\rho)\big\|_{\mathrm{tr}}=\left\lVert\sum_k p_k \mathcal{Q}_{P,Q,\mu,\nu}\Big(\rho^A_k\otimes\rho^B_k\Big)\right\rVert_{\mathrm{tr}}
	\leq\sum_k p_k \big\|\mathcal{Q}_{P,Q,\mu,\nu}(\rho^A_k\otimes\rho^B_k)\big\|_{\mathrm{tr}}
	=\mathfrak{M}(P,Q,\mu,\nu),
\end{equation*}
which completes the proof.
\end{proof}

\begin{corollary}\label{cor1}
If $\rho\in\mathcal{H}_A\otimes\mathcal{H}_B$ is separable, then
\begin{equation}\label{c1}
\big\|\mathcal{Q}_{P,Q,\mu,\nu}(\rho)\big\|_{\mathrm{tr}}
\leq\sqrt{\|P\|_{\operatorname{SH}}^2+|\nu|^2}\ \sqrt{\|Q\|_{\mathrm{SH}}^2+|\mu|^2}.
\end{equation}
\end{corollary}

\begin{proof}
Let $X=P^TP+\nu\nu^T$ and $Y=Q^TQ+\mu\mu^T$, which are positive semi-definite. Using the  inequality $\|AB\|_{\mathrm{tr}}\leq\|A\|_{\text{SH}}\|B\|_{\text{SH}}$, we have
\begin{equation*}
	\mathfrak{M}(P,Q,\mu,\nu)=\big\|X^{1/2}Y^{1/2}\big\|_{\mathrm{tr}}
	\leq\big\|X^{1/2}\big\|_{\text{SH}}\big\|Y^{1/2}\big\|_{\text{SH}},
\end{equation*} where $\|A\|_{\text{SH}}=\sqrt{Tr(A^{\dagger}A)}$ represents the Schmidt–Hilbert norm of $A$.
Since $X$ is positive semi-definite, $\|X^{1/2}\|_{\text{SH}}^2=\operatorname{Tr}(X)=\|P\|_{\text{SH}}^2+|\nu|^2$, and similarly $\|Y^{1/2}\|_{\text{SH}}^2=\|Q\|_{\text{SH}}^2+|\mu|^2$, which gives directly to (\ref{c1}).
\end{proof}

\begin{example}\label{ex:1}
Let
\begin{equation*}
	\rho_{a}=\frac{1}{1+7a}
	\begin{pmatrix}
		a & 0 & 0 & 0 & 0 & a & 0 & 0 \\
		0 & a & 0 & 0 & 0 & 0 & a & 0 \\
		0 & 0 & a & 0 & 0 & 0 & 0 & a \\
		0 & 0 & 0 & a & 0 & 0 & 0 & 0 \\
		0 & 0 & 0 & 0 & \frac{1+a}{2} & 0 & 0 & \frac{\sqrt{1-a^2}}{2} \\
		a & 0 & 0 & 0 & 0 & a & 0 & 0\\
		0 & a & 0 & 0 & 0 & 0 & a & 0\\
		0 & 0 & a & 0 & \frac{\sqrt{1-a^2}}{2} & 0 & 0 & \frac{1+a}{2} \\
	\end{pmatrix},\quad 0 < a < 1.
\end{equation*}
Consider the state $\rho_{x}=\left(1-x\right)\rho_{a}+x\left|\xi \right\rangle\left\langle\xi\right|$, where $\left|\xi\right\rangle=\frac{1}{\sqrt{2}}(|00\rangle+|11\rangle)$.

Set $a=0.9$ and take
\begin{equation*}
	P=
	\begin{pmatrix}
		9.28 & 10.51 \\
		10.44 & 10.65
	\end{pmatrix},~~~
	Q=
	\begin{pmatrix}
		10.01 & 12.75 \\
		13.03 & 13.54
	\end{pmatrix},~~~
	\mu=
	\begin{pmatrix}
		\frac{1}{\sqrt{2}} \\[6pt]
		\frac{1}{\sqrt{2}}
	\end{pmatrix},~~~
	\nu=
	\begin{pmatrix}
		\frac{1}{\sqrt{2}} \\[6pt]
		\frac{1}{\sqrt{2}}
	\end{pmatrix}.
\end{equation*}
By direct calculation we get from Theorem \ref{thm1} that $\rho_x$ is entangled for $0.223231 \leq x \leq 1$. From Theorem 1 in Ref.~\cite{LuQIP134}, taking $\boldsymbol{a}=(11.9967,12.9195,11.6808,12.1705,11.4476)^T$ and $\boldsymbol{b}=(12.5025,11.5102,12.0119,12.3982,12.7818)^T$ in the matrix (\ref{realignment03}), one has that $\rho_x$ is entangled for $0.232959 \leq x \leq 1$. For $\alpha=10$, $\beta=12$, and $l=5$ in the matrix (\ref{realignment02}), Theorem 1 in Ref.~\cite{SunPRA012415} gives rise to that $\rho_x$ is entangled for $0.223267 \leq x \leq 1$. Taking $\alpha=10$ and $\beta=12$ in the matrix (\ref{realignment01}), from Theorem 1 in Ref.~\cite{Shi2023family}, one has that $\rho_x$ is entangled for $0.223276 \leq x \leq 1$. The original realignment criterion of Ref.~\cite{Chen2003matrix} only detects entanglement for $0.252746 \leq x \leq 1$. It is seen that our theorem detects not only the entangled states identifiable by the related existing criteria, but also entangled states that those criteria fail to detect. These results are presented in Table~\ref{tabeg1} as follows.

\begin{table}[htbp]
  \centering
  \caption{Comparison of entanglement detection ranges of $\rho_x$ by different entanglement criteria}
  \label{tabeg1}
  \begin{tabular}{lc}
    \toprule
    Criterion & Range of $x$ for detected entanglement \\
    \midrule
    Theorem \ref{thm1} & $0.223231 \leq x \leq 1$ \\
    Theorem 1 in Ref.~\cite{LuQIP134} & $0.232959 \leq x \leq 1$ \\
    Theorem 1 in Ref.~\cite{SunPRA012415} & $0.223267 \leq x \leq 1$ \\
    Theorem 1 in Ref.~\cite{Shi2023family} & $0.223276 \leq x \leq 1$ \\
    Realignment criterion in Ref.~\cite{Chen2003matrix} & $0.252746 \leq x \leq 1$ \\
    \bottomrule
  \end{tabular}
\end{table}
\end{example}

\begin{example}\label{ex:2}
Consider a mixture of the bound entangled state \cite{HorodeckiPLA333}
\begin{equation*}
	\rho_{t}=\frac{1}{1+8t}
	\begin{pmatrix}
		t & 0 & 0 & 0 & t & 0 & 0 & 0 & t \\
		0 & t & 0 & 0 & 0 & 0 & 0 & 0 & 0 \\
		0 & 0 & t & 0 & 0 & 0 & 0 & 0 & 0 \\
		0 & 0 & 0 & t & 0 & 0 & 0 & 0 & 0 \\
		t & 0 & 0 & 0 & t & 0 & 0 & 0 &  t \\
		0 & 0 & 0 & 0 & 0 & t & 0 & 0 & 0 \\
		0 & 0 & 0 & 0 & 0 & 0 & \frac{1+t}{2} & 0 & \frac{\sqrt{1-t^2}}{2}\\
		0 & 0 & 0 & 0 & 0 & 0 & 0 & t & 0\\
		t & 0 & 0 & 0 & t & 0 & \frac{\sqrt{1-t^2}}{2} & 0 & \frac{1+t}{2}\\
	\end{pmatrix}
\end{equation*}
and the $9\times9$ identity matrix $I_9$,
\begin{equation*}
\rho_{p}=\frac{1-p}{9} I_{9} + p\rho_{t}.
\end{equation*}
We take
\begin{equation*}
P=
\begin{pmatrix}
3.511157 & 1.621537 \\
0.8302 & 2.934043 \\
1.626937 & 2.083113 \\
2.058475 & 2.54207 \\
2.053231 & 1.890969 \\
2.197353 & 3.218868 \\
\end{pmatrix},~~~
Q=
\begin{pmatrix}
3.317891 & 1.579999 \\
1.533877 & 1.714491 \\
2.860596 & 3.503382 \\
2.042115 & 3.078828 \\
0.519943 & 3.29257 \\
1.828804 & 0.217855
\end{pmatrix},~~~
\mu=
\begin{pmatrix}
\frac{1}{\sqrt{2}} \\[6pt]
\frac{1}{\sqrt{2}}
\end{pmatrix},~~~
\nu=
\begin{pmatrix}
\frac{1}{\sqrt{2}} \\[6pt]
\frac{1}{\sqrt{2}}
\end{pmatrix}.
\end{equation*}
Table \ref{tabeg2} compares the results from our Theorem \ref{thm1} and the corresponding ones from Refs.~\cite{Chen2003matrix,Shi2023family,SunPRA012415,LuQIP134} for different values of $t$. From Table \ref{tabeg2}, it is seen that our Theorem \ref{thm1} detects more entangled states than other criteria.

\begin{table}[htbp]
\caption{Ranges of $p$ for which the different criteria detect entanglement of $\rho_p$ for selected values of $t$}
\label{tabeg2}
\begin{tabular}{cccccc}
\toprule
$t$ & Theorem 1 in Ref.~\cite{Chen2003matrix} & Theorem 1 in Ref.~\cite{Shi2023family} & Theorem 1 in Ref.~\cite{SunPRA012415}  & Theorem 1 in Ref.~\cite{LuQIP134} & Our Theorem \ref{thm1} \\
\midrule
0.2 & $0.9955 \leq p \leq 1$ & $0.9943 \leq p \leq 1$ & $0.99408  \leq p \leq 1$ & $0.99406  \leq p \leq 1$ & $0.994043 \leq p \leq 1$ \\
0.4 & $0.9960 \leq p \leq 1$ & $0.9948  \leq p \leq 1$ & $0.99463  \leq p \leq 1$ & $0.99462  \leq p \leq 1$ & $0.994607 \leq p \leq 1$ \\	
0.6 & $0.9972 \leq p \leq 1$ & $0.9964  \leq p \leq 1$ & $0.99627  \leq p \leq 1$ & $0.99626 \leq p \leq 1$ & $0.996251 \leq p \leq 1$ \\	
0.8 & $0.9986 \leq p \leq 1$ & $0.9982  \leq p \leq 1$ & $0.998127  \leq p \leq 1$ & $0.998125  \leq p \leq 1$ & $0.9981233 \leq p \leq 1$ \\	
0.9 & $0.9993 \leq p \leq 1$ & $0.9991  \leq p \leq 1$ & $0.99906893  \leq p \leq 1$ & $0.99906843  \leq p \leq 1$ & $0.999067726 \leq p \leq 1$ \\
\bottomrule
\end{tabular}
\end{table}

Here, $\alpha=2$ and $\beta=2$ in the third column correspond to the matrix (\ref{realignment01}); $\alpha=2$, $\beta=2$, and $l=10$ in the fourth column correspond to the matrix (\ref{realignment02}); $\mathbf{a}=\left(\frac{17}{8},2,2,2,2,2,2,2,2,2\right)^T$ and $\mathbf{b}=\left(\frac{21}{10},2,2,2,2,2,2,2,2,2\right)^T$ in the fifth column correspond to the matrix (\ref{realignment03}).
\end{example}

\subsection{Lower bounds of concurrence and  convex-roof extended negativity (CREN)}

The concurrence of a pure state \(|\phi\rangle\in\mathcal{H}_A \otimes \mathcal{H}_B\) is defined as \cite{RungtaPRA042315}
\begin{equation}\label{concurrence}
C(|\phi\rangle)=\sqrt{2\big[1-\operatorname{Tr}(\rho_A^2)\big]},
\end{equation}
where \(\rho_A=\operatorname{Tr}_B\left(|\phi\rangle\langle\phi|\right)\) is the reduced density matrix of subsystem $A$. The concurrence of a mixed state \(\rho\) is defined via the convex-roof extension,
\begin{equation}
C(\rho)=\min_{\{p_k,|\phi_k\rangle\}}\sum_k p_k C(|\phi_k\rangle),
\end{equation}
where the minimization is performed over all ensemble decompositions \(\rho=\sum\limits_k p_k|\phi_k\rangle\langle\phi_k|\) with \(p_k\ge0\) and \(\sum\limits_k p_k=1\).

The CREN of a pure state $|\phi\rangle\in\mathcal{H}_A \otimes \mathcal{H}_B$ is defined by \cite{LeePRA062304}
\begin{equation}
N(|\phi\rangle)=\frac{\big\|T_B(|\phi\rangle\langle\phi|)\big\|_{\mathrm{tr}}-1}{d-1},
\end{equation}
where $d=\min\{d_A,d_B\}$ and $T_B$ denotes the partial transpose with respect to the subsystem $B$. The CREN of a mixed state $\rho$ is defined via the convex roof extension,
\begin{equation}
N(\rho)=\min_{\{p_k,|\phi_k\rangle\}}\sum_k p_k N(|\phi_k\rangle),
\end{equation}
where the minimum is taken over all possible pure state decompositions of $\rho$.

To derive lower bounds of concurrence and CREN, we introduce the following lemma, with its proof provided in Appendix \ref{appendixB}.

\begin{lemma}\label{lem2}
Let $|\phi\rangle$ be a pure bipartite state with Schmidt decomposition $|\phi\rangle=\sum\limits_{k=0}^{d-1}\sqrt{\lambda_k}|k_A k_B\rangle$, where $d=\min\{d_A, d_B\}$. We have 
\begin{enumerate}[label=(\arabic*)]
	\item $\displaystyle\big\|\mathcal{Q}_{P,Q,\mu,\nu}(|\phi\rangle\langle\phi|)\big\|_{\mathrm{tr}} \leq\mathfrak{M}(P, Q, \mu, \nu) + 2\sum_{0 \leq i < j \leq d-1}\sqrt{\lambda_i \lambda_j}$;
	\item $\displaystyle\big\|\mathcal{Q}_{P,Q,\mu,\nu}(|\phi\rangle\langle\phi|)\big\|_{\mathrm{tr}}\leq\mathfrak{M}(P,Q,\mu,\nu)+(d-1)$.
\end{enumerate}
\end{lemma}

According to the lemma above, we obtain the following theorems.

\begin{theorem}\label{thm2}
The concurrence of a bipartite state $\rho\in\mathcal{H}_A\otimes\mathcal{H}_B$ satisfies that
\begin{equation}
C(\rho)\geq\frac{\sqrt{2}}{\sqrt{(d-1)d\,}}\Big(\big\|\mathcal{Q}_{P,Q,\mu,\nu}
(\rho)\big\|_{\mathrm{tr}}-\mathfrak{M}(P,Q,\mu,\nu)\Big),
\end{equation}
where $d=\min\{d_A,d_B\}$.
\end{theorem}

\begin{proof}
For any pure state $|\phi\rangle=\sum\limits_{k=0}^{d-1}\sqrt{\lambda_k}|k_A k_B\rangle$, it is known that \cite{ChenAlbeverioFei}
\begin{equation*}
	C(|\phi\rangle)^2\geq\frac{8}{(d-1)d\,}\left(\sum_{0\leq i<j\leq d-1}\sqrt{\lambda_i\lambda_j}\right)^2.
\end{equation*}
From Lemma \ref{lem2}(i), $\sum\limits_{i<j}\sqrt{\lambda_i\lambda_{j}}\geq\frac{1}{2}
\big(\|\mathcal{Q}_{P,Q,\mu,\nu}(|\phi\rangle\langle\phi|)\|_{\mathrm{tr}}-\mathfrak{M}(P,Q,\mu,\nu)\big)$. Hence,
\begin{equation*}
C(|\phi\rangle)\geq\frac{2\sqrt2}{\sqrt{(d-1)d\,}}\sum_{i<j}\sqrt{\lambda_i\lambda_j}
\geq\frac{\sqrt2}{\sqrt{(d-1)d\,}}\Big(\big\|\mathcal{Q}_{P,Q,\mu,\nu}(|\phi\rangle\langle\phi|)\big\|_{\mathrm{tr}}-\mathfrak{M}(P,Q,\mu,\nu)\Big).
\end{equation*}

Let $\{p_k,|\phi_k\rangle\}$ be the optimal pure-state decomposition of $\rho$ such that $C(\rho)=\sum\limits_k p_k C(|\phi_k\rangle)$. Using Lemma \ref{lem1}(1) and the triangle inequality, we have
\begin{equation*}
\begin{split}
C(\rho)=\sum_k p_k C(|\phi_k\rangle)
&\geq\frac{\sqrt2}{\sqrt{(d-1)d\,}}\sum_k p_k\Big(\|\mathcal{Q}_{P,Q,\mu,\nu}(|\phi_k\rangle\langle\phi_k|)\|_{\mathrm{tr}}-\mathfrak{M}(P,Q,\mu,\nu)\Big)\\&\geq\frac{\sqrt2}{\sqrt{(d-1)d\,}}\Big(\big\|\mathcal{Q}_{P,Q,\mu,\nu}(\rho)\big\|_{\mathrm{tr}}-\mathfrak{M}(P,Q,\mu,\nu)\Big),
\end{split}
\end{equation*}
which completes the proof.
\end{proof}

\begin{theorem}\label{thm3}
For the CREN of any state $\rho\in\mathcal{H}_A\otimes\mathcal{H}_B$, we have
\begin{equation}
N(\rho)\geq\frac{\big\|\mathcal{Q}_{P,Q,\mu,\nu}(\rho)\big\|_{\mathrm{tr}}-\mathfrak{M}(P,Q,\mu,\nu)}{d-1},
\end{equation}
where $d=\min\{d_A,d_B\}$.
\end{theorem}

\begin{proof}
For any pure state $|\phi\rangle$ with Schmidt coefficients $\lambda_k$, one has \cite{LeePRA062304}
\begin{equation*}
	N(|\phi\rangle)=\frac{2}{d-1}\sum_{0\leq i<j\leq d-1}\sqrt{\lambda_i\lambda_j}.
\end{equation*}
From Lemma \ref{lem2}(i), we can get
\begin{equation*}
	N(|\phi_k\rangle\langle\phi_k|)\geq\frac{\big\|\mathcal{Q}_{P,Q,\mu,\nu}(|\phi\rangle\langle\phi|)\big\|_{\mathrm{tr}}-\mathfrak{M}(P,Q,\mu,\nu)}{d-1}.
\end{equation*}
Let $\{p_k,|\phi_k\rangle\}$ be the optimal decomposition of $\rho$ such that $N(\rho)=\sum\limits_k p_k N(|\phi_k\rangle)$. Using Lemma \ref{lem1}(1) and the triangle inequality, we obtain
\begin{equation*}
	N(\rho)=\sum_k p_k N(|\phi_k\rangle)\geq \sum_k p_k\frac{\big\|\mathcal{Q}_{P,Q,\mu,\nu}(|\phi\rangle_k\langle_k\phi|)\big\|_{\mathrm{tr}}-\mathfrak{M}(P,Q,\mu,\nu)}{d-1}\geq\frac{\big\|\mathcal{Q}_{P,Q,\mu,\nu}(\rho)\big\|_{\mathrm{tr}}-\mathfrak{M}(P,Q,\mu,\nu)}{d-1}.
\end{equation*}
which completes the proof.
\end{proof}

\begin{example}
The following $3 \times 3$ PPT entangled state was introduced in Ref. \cite{BennettPRL5385},
\begin{equation*}
\rho=\frac{1}{4}\left(I_{9}-\sum_{k=0}^{4}
\left|\phi_{k}\right\rangle\left\langle\phi_{k}\right|\right),
\end{equation*}
where
\begin{equation*}
	\begin{aligned}
		& \left|\phi_{0}\right\rangle=|0\rangle\otimes\frac{1}{\sqrt{2}}(|0\rangle-|1\rangle), \\
		& \left|\phi_{1}\right\rangle=\frac{1}{\sqrt{2}}(|0\rangle-|1\rangle)\otimes|2\rangle, \\
		& \left|\phi_{2}\right\rangle=|2\rangle\otimes\frac{1}{\sqrt{2}}(|1\rangle-|2\rangle), \\
		& \left|\phi_{3}\right\rangle=\frac{1}{\sqrt{2}}(|1\rangle-|2\rangle)\otimes|0\rangle, \\
		& \left|\phi_{4}\right\rangle=\frac{1}{\sqrt{3}}(|0\rangle+|1\rangle+|2\rangle)\otimes\frac{1}{\sqrt{3}}(|0\rangle+|1\rangle+|2\rangle).
	\end{aligned}
\end{equation*}
We consider the mixture of $\rho$ with white noise,
\begin{equation*}
\rho_y = \frac{1-y}{9}I_9 + y\rho.
\end{equation*}

For concurrence, we take
\begin{equation*}
P=
\begin{pmatrix}
0.23 & 1.7 & 1.61 \\
1.57 & 2.65 & 0.85 \\
1.74 & 0.82 & 1.7 \\
1.43 & 1.8 &  -1.2
\end{pmatrix},\quad
Q=
\begin{pmatrix}
1.71 & 1.19 & 1.27 \\
1.53 & 1.93 & 1.68 \\
0.58 & 1.5 & 0.8 \\
0.54 & 1.33 & 1.2
\end{pmatrix},\quad
\mu=
\begin{pmatrix}
\frac{1}{\sqrt{3}} \\[6pt]
\frac{1}{\sqrt{3}} \\[6pt]
\frac{1}{\sqrt{3}}
\end{pmatrix},\quad
\nu=
\begin{pmatrix}
\frac{1}{\sqrt{3}} \\[6pt]
\frac{1}{\sqrt{3}} \\[6pt]
\frac{1}{\sqrt{3}}
\end{pmatrix}.
\end{equation*}
FIG.\ref{lower bounds of concurrence} shows the lower bound of the concurrence. The solid blue line is the lower bound of concurrence from our Theorem \ref{thm2}. The dashed orange line in FIG.\ref{lower bounds of concurrence} is the lower bound of concurrence from Theorem 2 in Ref.~\cite{LuQIP134} with $\mathbf{a}=\mathbf{b}=(1,1,1.2)^T$ in the matrix (\ref{realignment03}). The dot-dashed green line in FIG.\ref{lower bounds of concurrence} is the lower bound of concurrence from Theorem 2 in Ref.~\cite{SunPRA012415} with $\alpha=\beta=1$ and $l=2$ in the matrix (\ref{realignment02}). The dotted pink line in FIG.\ref{lower bounds of concurrence} is the lower bound of concurrence from Theorem 6 in Ref.~\cite{Shi2023family} with $\alpha=\beta=1$ in the matrix (\ref{realignment01}).  Clearly, our bound in Theorem \ref{thm2} outperforms those given in Refs.~\cite{LuQIP134,SunPRA012415,Shi2023family}.
\begin{figure}[htbp] 
    \centering
    \includegraphics[width=0.7\textwidth]{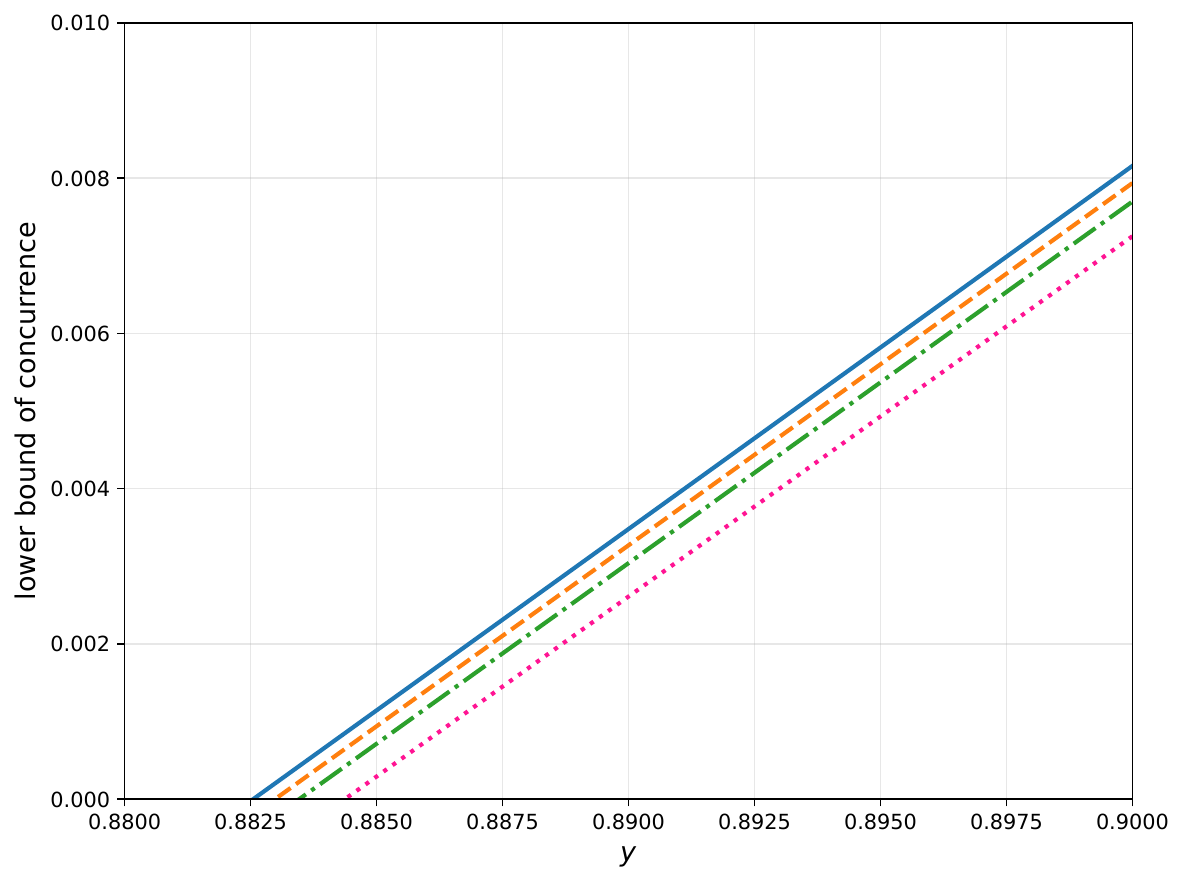}
    \caption{Lower bounds of concurrence for $\rho_y$. The solid blue line is our result from Theorem \ref{thm2}. The dashed orange, dot-dashed green, dotted pink lines stand for the bounds from Theorem 2 in Ref.~\cite{LuQIP134}, Theorem 2 in Ref.~\cite{SunPRA012415}, Theorem 6 in Ref.~\cite{Shi2023family}, respectively.}
    \label{lower bounds of concurrence}
\end{figure}

For CREN, we take
\begin{equation*}
P=
\begin{pmatrix}
0.62 & 1.01 & 1.32 & 2.36 & 0.32 \\
2.22 & 2.32 & 1.15 & 1.28 & 2.8 \\
0.64 & 1.61 & 1.49 & 0.31 & 0.96
\end{pmatrix},\quad
Q=
\begin{pmatrix}
0.31 & 0.07 & 1.73 & 0.48 & 3.41 \\
0.82 & 2 & 2.18 & 2.51 & 1.06 \\
2.87 & 1.49 & 0.19 & 1 & 1.54
\end{pmatrix},\quad
\mu=
\begin{pmatrix}
\frac{1}{\sqrt{5}} \\[6pt]
\frac{1}{\sqrt{5}} \\[6pt]
\frac{1}{\sqrt{5}} \\[6pt]
\frac{1}{\sqrt{5}} \\[6pt]
\frac{1}{\sqrt{5}}
\end{pmatrix},\quad
\nu=
\begin{pmatrix}
\frac{1}{\sqrt{5}} \\[6pt]
\frac{1}{\sqrt{5}} \\[6pt]
\frac{1}{\sqrt{5}} \\[6pt]
\frac{1}{\sqrt{5}} \\[6pt]
\frac{1}{\sqrt{5}}
\end{pmatrix}.
\end{equation*}
The solid blue line in FIG.\ref{lower bounds of CREN} is the lower bound of CREN from our Theorem \ref{thm3}. The dashed orange line in FIG.\ref{lower bounds of CREN} is the lower bound of CREN from Theorem 3 in Ref.~\cite{LuQIP134} with $\mathbf{a}=\mathbf{b}=(1,1,1.2)^T$ in the matrix (\ref{realignment03}). The dot-dashed green line in FIG.\ref{lower bounds of CREN} is the lower bound of CREN from Theorem 3 in Ref.~\cite{SunPRA012415} with $\alpha=\beta=1$ and $l=2$ in the matrix (\ref{realignment02}). The dotted pink line in FIG.\ref{lower bounds of CREN} is the lower bound of CREN from Theorem 7 in Ref.~\cite{Shi2023family} with $\alpha=\beta=1$ in the matrix (\ref{realignment01}). Clearly, our bound in Theorem \ref{thm3} outperforms those given in   Refs.~\cite{LuQIP134,SunPRA012415,Shi2023family}.
\end{example}
\begin{figure}[htbp] 
    \centering
    \includegraphics[width=0.7\textwidth]{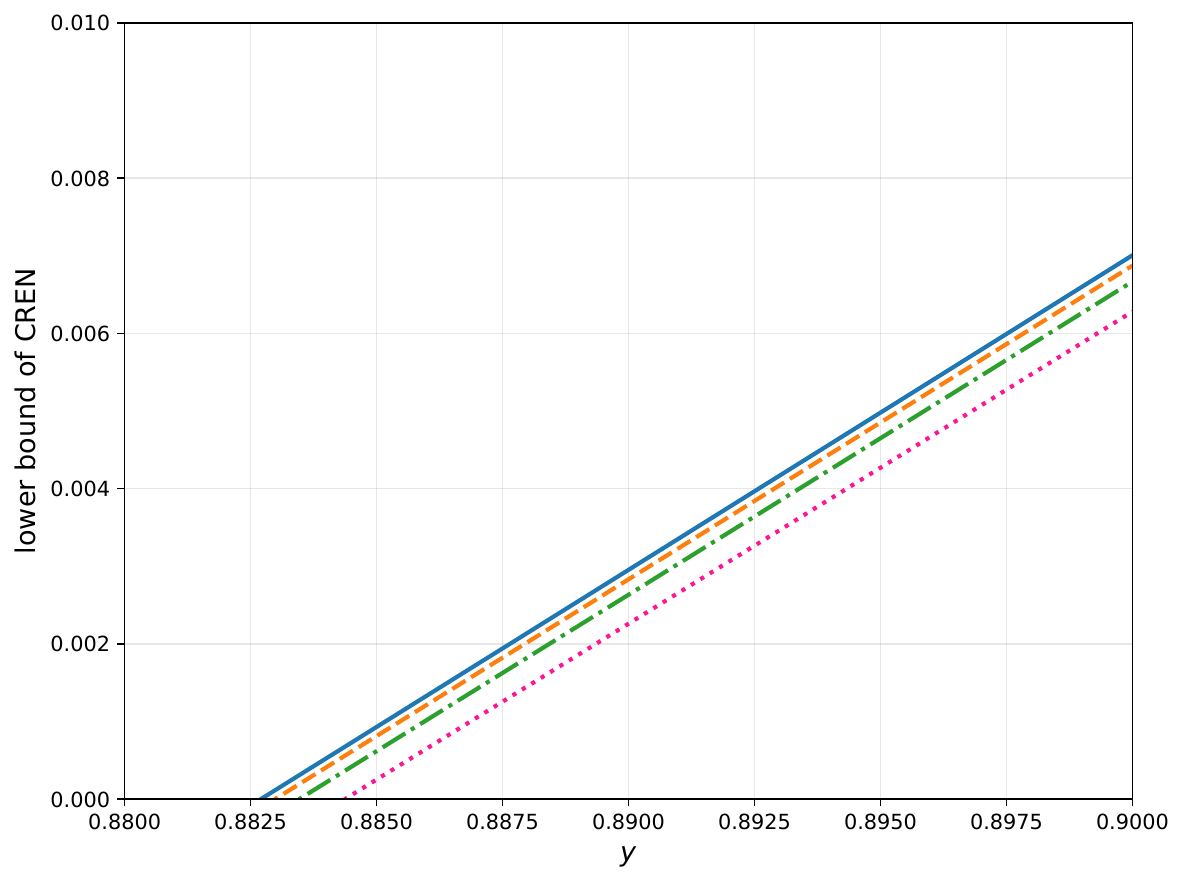}
    \caption{Lower bounds of CREN for $\rho_y$. The solid blue line is our result from Theorem \ref{thm3}. The dashed orange, dot-dashed green, dotted pink lines stand for the bounds from Theorem 3 in Ref.~\cite{LuQIP134}, Theorem 3 in Ref.~\cite{SunPRA012415}, Theorem 7 in Ref.~\cite{Shi2023family}, respectively.}
    \label{lower bounds of CREN}
\end{figure}

\section{Detection and quantification of multipartite entanglement}\label{sectIII}

\subsection{Separability criteria for multipartite quantum states}

We first consider the tripartite case. Denote the three bipartitions of a tripartite state $\rho\in\mathcal{H}_A\otimes\mathcal{H}_B\otimes\mathcal{H}_C$ as $A|BC$, $B|AC$ and $C|AB$. If $\rho$ is biseparable, then it can be written as a convex combination of pure states that are separable under these bipartitions, i.e.
\begin{equation}
	\rho=\sum_i p_i|\phi_i\rangle^{A|BC}\langle\phi_i|+\sum_j p_j|\phi_j\rangle^{B|AC}\langle\phi_j|+\sum_k p_k|\phi_k\rangle^{C|AB}\langle\phi_k|,
\end{equation}
where $p_i, p_j, p_k \geq 0$, with $\sum\limits_i p_i+\sum\limits_j p_j+\sum\limits_k p_k=1$. Otherwise, $\rho$ is called genuinely tripartite entangled. Define
\begin{equation}\label{tripartite concurrence}
	\mathscr{Q}_{P,Q,\mu,\nu}(\rho)=\frac{1}{3}\Big(\big\|\mathcal{Q}_{P,Q,\mu,\nu}^{A|BC}(\rho)\big\|_{\mathrm{tr}}+\big\|\mathcal{Q}_{P,Q,\mu,\nu}^{B|AC}(\rho)\big\|_{\mathrm{tr}}+\big\|\mathcal{Q}_{P,Q,\mu,\nu}^{C|AB}(\rho)\big\|_{\mathrm{tr}}\Big),
\end{equation}
where $\mathcal{Q}^{i|jk}_{P,Q,\mu,\nu}$ stands for the matrix (\ref{realignment matrix}) under the bipartition $i$ and $jk$, $i \neq j \neq k \in \{A,B,C\}$.

\begin{theorem}\label{thm4}
If a tripartite state $\rho\in\mathbb{C}^{d}\otimes\mathbb{C}^{d}\otimes\mathbb{C}^{d}$ is biseparable, then
\begin{equation*}
	\mathscr{Q}_{P,Q,\mu,\nu}(\rho)\leq\mathfrak{M}(P,Q,\mu,\nu)+\frac{2(d-1)}{3}.
\end{equation*}
\end{theorem}

\begin{proof}
	Let $|\phi\rangle\in\mathbb{C}^d\otimes\mathbb{C}^d\otimes\mathbb{C}^d$ be biseparable under bipartition $A|BC$, i.e., $|\phi\rangle=|\phi_A\rangle\otimes|\phi_{BC}\rangle$. Then $\|\mathcal{Q}_{P,Q,\mu,\nu}^{A|BC}(|\phi\rangle\langle\phi|)\|_{\mathrm{tr}}\leq\mathfrak{M}(P,Q,\mu,\nu)$ by Theorem \ref{thm1}, while Lemma \ref{lem2}(ii) gives $\|\mathcal{Q}_{P,Q,\mu,\nu}^{B|AC}(|\phi\rangle\langle\phi|)\|_{\mathrm{tr}}\leq \mathfrak{M}(P,Q,\mu,\nu)+(d-1)$ and $\|\mathcal{Q}_{P,Q,\mu,\nu}^{C|AB}(|\phi\rangle\langle\phi|)\|_{\mathrm{tr}}\leq\mathfrak{M}(P,Q,\mu,\nu)+(d-1)$. Hence, from Eq.~(\ref{tripartite concurrence}) we have
	\begin{equation*}
		\mathscr{Q}_{P,Q,\mu,\nu}(|\phi\rangle\langle\phi|)\leq\frac{1}{3}\Big(\mathfrak{M}(P,Q,\mu,\nu)+\big(\mathfrak{M}(P,Q,\mu,\nu)+d-1\big)+\big(\mathfrak{M}(P,Q,\mu,\nu)+d-1\big)\Big)=\mathfrak{M}(P,Q,\mu,\nu)+\frac{2(d-1)}{3}.
	\end{equation*}
	The same bound holds for biseparable states under the bipartitions $B|AC$ and $C|AB$. For a biseparable mixed state $\rho=\sum_kp_k|\phi_k\rangle\langle\phi_k|$, we have
\begin{equation*}
	\mathscr{Q}_{P,Q,\mu,\nu}(\rho)\leq\sum_k p_k \mathscr{Q}_{P,Q,\mu,\nu}(|\phi_k\rangle\langle\phi_k|),
\end{equation*}
which implies
\begin{equation*}
	\mathscr{Q}_{P,Q,\mu,\nu}(\rho)\leq\sum_k p_k \Big(\mathfrak{M}(P,Q,\mu,\nu)+\frac{2(d-1)}{3}\Big)=\mathfrak{M}(P,Q,\mu,\nu)+\frac{2(d-1)}{3}.
\end{equation*}
\end{proof}

\begin{example}
Consider the state $\rho_{\operatorname{W}}^{q}$ in $\mathbb{C}^3 \otimes \mathbb{C}^3 \otimes \mathbb{C}^3$,
\begin{equation*}
	\rho_{\operatorname{W}}^{q}=\frac{1-q}{27} I_{27} + q\left|\phi_{\operatorname{W}}\right\rangle\left\langle\phi_{\operatorname{W}}\right|,
\end{equation*}
where $\left|\phi_{\operatorname{W}}\right\rangle=\frac{1}{\sqrt{6}}(|001\rangle+|010\rangle+|100\rangle+|112\rangle+|121\rangle+|211\rangle)$, and $I_{27}$ is the $27 \times 27$ identity matrix. Take
\begin{equation*}
P=
\begin{pmatrix}
10.65 & 11.57 \\
11.88 & 11.94 \\
12.42 & 11.69
\end{pmatrix},\quad
Q=
\begin{pmatrix}
12.45 & 12.81 \\
13.05 & 12.1 \\
11.62 & 12.45
\end{pmatrix},\quad
\mu=
\begin{pmatrix}
\frac{1}{\sqrt{2}} \\[6pt]
\frac{1}{\sqrt{2}}
\end{pmatrix},\quad
\nu=
\begin{pmatrix}
\frac{1}{\sqrt{2}} \\[6pt]
\frac{1}{\sqrt{2}}
\end{pmatrix}.
\end{equation*}
Theorem \ref{thm4} detects the genuine tripartite entanglement of $\rho_{\operatorname{W}}^q$ for $0.804647 \leq q \leq 1$. When $\mathbf{a}=(1,2)^T$ and $\mathbf{b}=(2,1)^T$ in the matrix (\ref{realignment03}), Theorem \ref{thm4} reduces to the Theorem 4 in Ref.~\cite{LuQIP134}, which detects the genuine tripartite entanglement of $\rho_{\operatorname{W}}^q$ for $0.805132 \leq q \leq 1$. When $\alpha=\beta=1$ and $l=2$ in the matrix (\ref{realignment02}), Theorem \ref{thm4} reduces to the Theorem 4 in Ref.~\cite{SunPRA012415}, which detects the genuine tripartite entanglement of $\rho_{\operatorname{W}}^q$ for $0.805211 \leq q \leq 1$. And when $\alpha=\beta=1$ in the matrix (\ref{realignment01}), Theorem \ref{thm4} reduces to the Theorem 2 in Ref.~\cite{Qi2024}, which detects the genuine tripartite entanglement of $\rho_{\operatorname{W}}^q$ for $0.805318 \leq q \leq 1$. Obviously, our Theorem \ref{thm4} is more effective in detecting the genuine tripartite entanglement. These results are presented in Table~\ref{tabeg4} as follows.

\begin{table}[htbp]
  \centering
  \caption{Comparison of genuine tripartite entanglement detection ranges of $\rho_{\operatorname{W}}^q$ under different criteria}
  \label{tabeg4}
  \begin{tabular}{lc}
    \toprule
    Criterion & Range of $q$ for detected genuine entanglement \\
    \midrule
    Theorem \ref{thm4} & $0.804647 \leq q \leq 1$ \\
    Theorem 4 in Ref.~\cite{LuQIP134} & $0.805132 \leq q \leq 1$ \\
    Theorem 4 in Ref.~\cite{SunPRA012415} & $0.805211 \leq q \leq 1$ \\
    Theorem 2 in Ref.~\cite{Qi2024} & $0.805318 \leq q \leq 1$ \\
    \bottomrule
  \end{tabular}
\end{table}
\end{example}

Next we consider the full separability of general multipartite states. Any multipartite state $\rho\in\mathbb{C}^{d_{1}} \otimes \mathbb{C}^{d_{2}} \otimes \cdots \otimes \mathbb{C}^{d_{n}}$ can be written as
$\rho=\sum_{k} Y^{1}_{k} \otimes Y^{2}_{k} \otimes \cdots \otimes Y^{n}_{k}$,
where $Y^{j}_{k} \in \mathbb{C}^{d_{j} \times d_{j}}$, $j=1,2, \cdots, n$. We define
\begin{equation}\label{eq:13}
\mathcal{Q}\mathcal{R}_{P_q, \cdots, P_n, \mu_{q}, \cdots, \mu_{n}}^{q}(\rho)=\sum_{k} Y^{1}_{k} \otimes \cdots \otimes Y^{q-1}_{k} \otimes \mathcal{Q}_{P_q, \cdots, P_n, \mu_{q}, \cdots, \mu_{n}}\Big[Y^{q}_{k}, \cdots, Y^{n}_{k}\Big],
\end{equation}
where $q=1,\cdots,n-1$, $P_i \in \mathbb{R}^{m \times r}$, $\operatorname{rank}(P_i)=r$, $\mu_i\in\mathbb{R}^{r \times 1}$, $i = q, \cdots, n$,
\begin{equation*}
\mathcal{Q}_{P_q, \cdots, P_n, \mu_{q}, \cdots, \mu_{n}}\Big[Y^{q}_{k}, \cdots, Y^{n}_k\Big]=
\left(\begin{array}{c}
P_q \\ [6pt]
\operatorname{vec}\left(Y^{q}_{k}\right)\mu_q^T
\end{array}\right)
\bigotimes_{j=0}^{n-q-1}
\begin{pmatrix}
P_{n-j}^T & \quad\mu_{n-j}\operatorname{vec}\left(Y^{n-j}_{k}\right)^{T}
\end{pmatrix}.
\end{equation*}
Denote $\mathfrak{M}(P_i,\mu_i):=\left\lVert\Big(P_i^TP_i+\mu_i\mu_i^T\Big)^{1/2}\right\rVert_{\mathrm{tr}}$. We have the following separability criterion for multipartite states.

\begin{theorem}\label{thm5}
If a multipartite state  $\rho\in\mathbb{C}^{d_{1}} \otimes \mathbb{C}^{d_{2}} \otimes \cdots \otimes \mathbb{C}^{d_{n}}$ is fully separable, then for any $1\leq q \leq n-1$,
\begin{equation}
\left\|\mathcal{Q} \mathcal{R}_{P_q, \cdots, P_n, \mu_{q}, \cdots, \mu_{n}}^{q}(\rho)\right\|_{\mathrm{tr}} \leq \prod_{i=q}^{n} \mathfrak{M}(P_i,\mu_i),
\end{equation}
where $\mathcal{Q} \mathcal{R}_{P_q, \cdots, P_n, \mu_{q}, \cdots, \mu_{n}}^{q}(\rho)$ is defined in Eq.~\eqref{eq:13}.
\end{theorem}

\begin{proof}
Since any fully separable state $\rho$ can be written as
$\rho=\sum_{k} p_{k} \rho^{1}_{k} \otimes \rho^{2}_{k} \otimes \cdots \otimes \rho^{n}_{k}$,
where $p_{k} \in[0,1]$ with $\sum\limits_{k} p_{k}=1$, we have
\begin{equation*}
	\begin{aligned}
		\left\|\mathcal{Q}\mathcal{R}_{P_q, \cdots, P_n, \mu_{q}, \cdots, \mu_{n}}^{ q}(\rho)\right\|_{\mathrm{tr}}&=\left\|\sum\limits_{k} p_{k} \rho^{1}_{k} \otimes \cdots \otimes \rho^{q-1}_{k} \otimes \mathcal{Q}_{P_q, \cdots, P_n, \mu_{q}, \cdots, \mu_{n}}\Big[\rho^{q}_{k}, \cdots, \rho^{n}_{k}\Big]\right\|_{\mathrm{tr}}  \\
		&\leq \sum\limits_{k}p_{k}\left\|\rho^{1}_{k} \otimes \cdots \otimes \rho^{q-1}_{k} \otimes \mathcal{Q}_{P_q, \cdots, P_n, \mu_{q}, \cdots, \mu_{n}}\Big[\rho^{q}_{k}, \cdots, \rho^{n}_{k}\Big]\right\|_{\mathrm{tr}}  \\
		&=\sum\limits_{k} p_{k}\operatorname{Tr}\left(\rho^{1}_{k}\right) \cdots\operatorname{Tr}\left(\rho^{q-1}_{k}\right)\prod_{i=q}^{n}\left\|\left(
		\begin{array}{c}
			P_i \\ [6pt]
			\operatorname{vec}\left(\rho^{i}_{k}\right)\mu_i^T
		\end{array}\right)\right\|_{\mathrm{tr}} \\
		&=\sum\limits_{k} p_{k} \prod_{i=q}^{n} \mathfrak{M}(P_i,\mu_i) \\
		&=\prod_{i=q}^{n} \mathfrak{M}(P_i,\mu_i),
	\end{aligned}
\end{equation*}
which proves Theorem \ref{thm5}.
\end{proof}

\subsection{Lower bounds of genuine-multipartite-entanglement concurrence for tripartite states}

The genuine-multipartite-entanglement (GME) concurrence of a pure state $|\phi\rangle\in\mathbb{C}^{d}\otimes\mathbb{C}^{d}\otimes\mathbb{C}^{d}$ is defined by \cite{MaGMEconcurrence}
\begin{equation}
	C_{\mathrm{GME}}(|\phi\rangle)=\min\left\{\sqrt{1-\operatorname{Tr}(\rho_1^2)},\ \sqrt{1-\operatorname{Tr}(\rho_2^2)},\ \sqrt{1-\operatorname{Tr}(\rho_3^2)}\,\right\},
\end{equation}
where $\rho_k$ is the reduced density matrix of the subsystem $k$. The GME concurrence of a mixed state $\rho$ is defined as
\begin{equation}
	C_{\mathrm{GME}}(\rho)=\min_{\{p_k,|\phi_k\rangle\}}\sum_k p_k C_{\mathrm{GME}}(|\phi_k\rangle),
\end{equation}
where the minimum is taken over all possible ensemble decompositions of $\rho=\sum\limits_k p_k |\phi_k\rangle\langle\phi_k|$, $p_k \geq 0$ with $\sum\limits_k p_k =1$. The following theorem gives a lower bound for the GME concurrence. Its proof is provided in Appendix \ref{appendixC}.

\begin{theorem}\label{thm6}
For any tripartite state $\rho\in\mathbb{C}^{d}\otimes\mathbb{C}^{d}\otimes\mathbb{C}^{d}$, we have
\begin{equation*}
	C_{\mathrm{GME}}(\rho)\ \geq\ \frac{1}{\sqrt{(d-1)d\,}}\Big(\mathscr{Q}_{P,Q,\mu,\nu}(\rho)-\mathfrak{M}(P,Q,\mu,\nu)-\frac{2(d-1)}{3}\Big),
\end{equation*}
where $\mathscr{Q}_{P,Q,\mu,\nu}(\rho)$ is defined in Eq.~(\ref{tripartite concurrence}).
\end{theorem}

\begin{example}
Consider the following state $\rho_{x}\in\mathbb{C}^2 \otimes \mathbb{C}^2 \otimes \mathbb{C}^2$,
\begin{equation*}
	\rho_{x}=(1-x)\left|\phi_{\operatorname{GHZ}}\right\rangle\left\langle\phi_{\operatorname{GHZ}}\right| + \frac{x}{8}I_{8}, \quad 0 \leq x \leq 1,
\end{equation*}
where $\left|\phi_{\operatorname{GHZ}}\right\rangle=\frac{1}{\sqrt{2}}(|000\rangle+|111\rangle)$, $I_{8}$ is the $8 \times 8$ identity matrix.
Let
\begin{equation*}
P=
\begin{pmatrix}
1.39 & -0.29 & 0.34 & 1 \\
0.19 & 0.93 & 1.65 & 1.58 \\
0.94 & 0.74 & 1.55 & 1.1
\end{pmatrix},\quad
Q=
\begin{pmatrix}
-0.37 & 0.15 & 3.28 & 2.03 \\
1.54 & 0.05 & 1 & 0.74 \\
-1.14 & 0.06 & 1.56 & 2.54
\end{pmatrix},\quad
\mu=
\begin{pmatrix}
\frac{1}{2} \\[4pt]
\frac{1}{2} \\[4pt]
\frac{1}{2} \\[4pt]
\frac{1}{2}
\end{pmatrix},\quad
\nu=
\begin{pmatrix}
\frac{1}{2} \\[4pt]
\frac{1}{2} \\[4pt]
\frac{1}{2} \\[4pt]
\frac{1}{2}
\end{pmatrix}.
\end{equation*}
 FIG.\ref{lower bound of GME concurrence} shows the lower bound of GME concurrence of $\rho_{x}$. The solid blue line is from our Theorem \ref{thm6}. It is seen that $\rho_x$ is genuine tripartite entangled for $0 \leq x \leq 0.193218$. The dashed orange line is from Theorem 6 in Ref.~\cite{LuQIP134} with $\mathbf{a}=(1,2)^T$ and $\mathbf{b}=(2,1)^T$ in the matrix (\ref{realignment03}), from which $\rho_x$ is genuine tripartite entangled for $0 \leq x \leq 0.193038$. The dot-dashed green line is from Theorem 6 in Ref.~\cite{SunPRA012415} with $\alpha=\beta=1$ and $l=2$ in the matrix (\ref{realignment02}), showing that $\rho_x$ is genuinely tripartite entangled for $0 \leq x \leq 0.192913$. The dotted pink line is from Theorem 3 in Ref.~\cite{Qi2024} with $\alpha=\beta=1$ in the matrix (\ref{realignment01}), from which $\rho_x$ is genuine tripartite entangled for $0 \leq x \leq 0.192758$. It is seen that our Theorem \ref{thm6} is more effective at detecting the genuine tripartite entanglement, yielding a larger lower bound on the GME concurrence for \(\rho_x\). These results are presented in Table~\ref{tabeg5} as follows.

\begin{table}[htbp]
  \centering
  \caption{Comparison of genuine tripartite entanglement detection ranges of $\rho_x$ by different criteria}
  \label{tabeg5}
  \begin{tabular}{lc}
    \toprule
    Criterion & Range of $q$ for detected genuine entanglement \\
    \midrule
    Theorem \ref{thm6} & $0 \leq x \leq 0.193218$ \\
    Theorem 6 in Ref.~\cite{LuQIP134} & $0 \leq x \leq 0.193038$ \\
    Theorem 6 in Ref.~\cite{SunPRA012415} & $0 \leq x \leq 0.192913$ \\
    Theorem 3 in Ref.~\cite{Qi2024} & $0 \leq x \leq 0.192758$\\
    \bottomrule
  \end{tabular}
\end{table}
\end{example}
\begin{figure}[htbp] 
    \centering
    \includegraphics[width=0.7\textwidth]{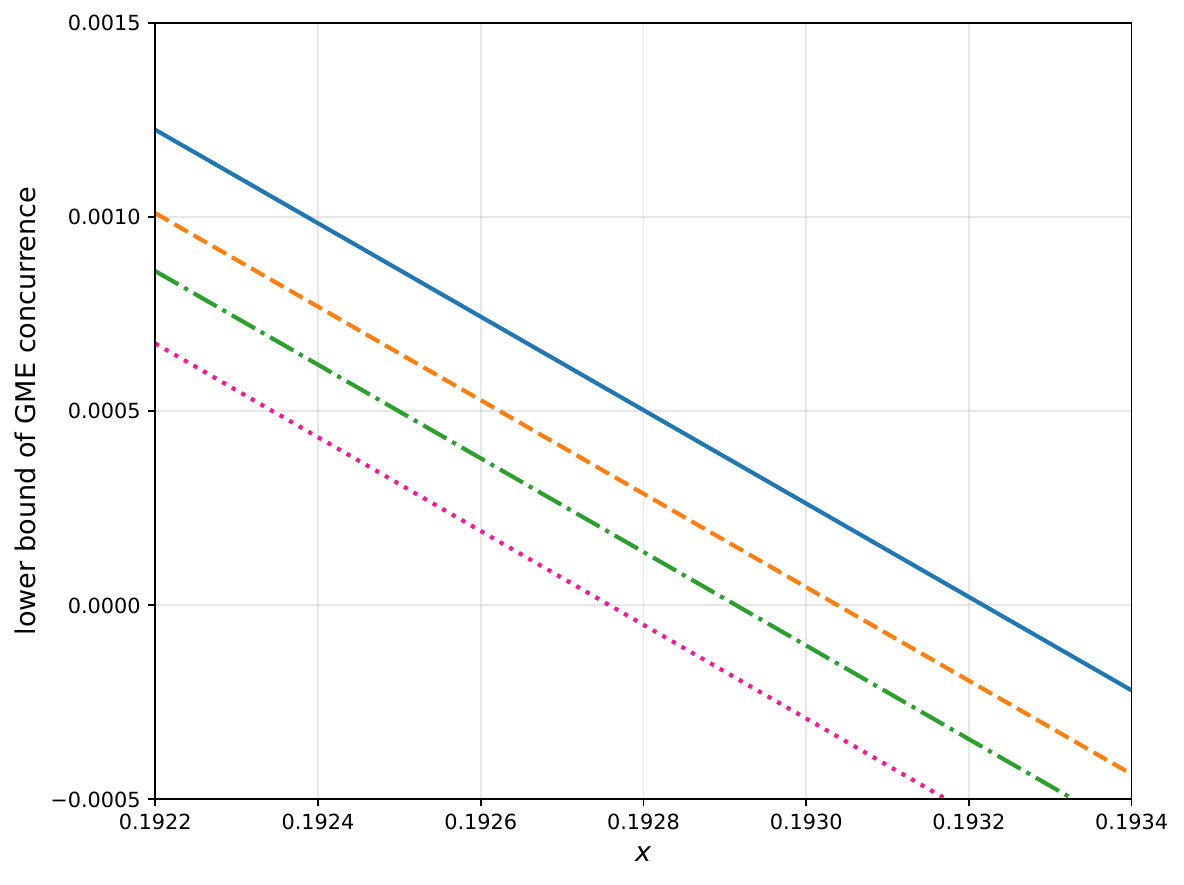}
    \caption{Lower bound of GME concurrence for $\rho_x$. The solid blue line is our result from Theorem \ref{thm6}. The dashed orange, dot-dashed green, dotted pink lines stand for the bounds from Theorem 6 in Ref.~\cite{LuQIP134}, Theorem 6 in Ref.~\cite{SunPRA012415}, Theorem 3 in Ref.~\cite{Qi2024}, respectively.}
    \label{lower bound of GME concurrence}
\end{figure}

\section{Conclusion}\label{sectIV}

We have presented new separability criteria by constructing matrices via the realignment of density matrices and the vectorization of reduced density matrices. The main novelty lies in the fact that the auxiliary block in the upper-left corner of the constructed matrix \(\mathcal{Q}_{P,Q,\mu,\nu}(\rho)\) can be an arbitrary rank-$r$ matrix \(PQ^T\), rather than a scalar or rank-one matrix adopted in prior works. In this setting, our matrix recovers those proposed by Shi and Sun, Sun \emph{et al.}, and Lu \emph{et al.} as special cases by setting \(r=1\). Using these criteria. we have derived lower bounds for the concurrence and convex-roof extended negativity of bipartite states. The criteria are further generalized to multipartite systems, yielding criteria for identifying genuine tripartite entanglement, full separability, and lower bounds on the genuine multipartite entanglement concurrence. Concrete examples illustrate that our criteria detect more entangled states and produce tighter lower bounds than the related existing realignment-based criteria.

Several problems merit further investigation. First, the free parameters P, Q, \(\mu\), and \(\nu\) in this work are selected via numerical optimization. It would be worthwhile to analytically characterize the optimal choices of these parameters, or to derive a parameter-independent optimized version of the criterion by taking the supremum over all admissible parameters, i.e., \(\mathcal{W}(\rho)=\sup_{P,Q,\mu,\nu}\Big[\big\|\mathcal{Q}_{P,Q,\mu,\nu}(\rho)\big\|_{\mathrm{tr}}-\mathfrak{M}(P,Q,\mu,\nu)\Big],\)
where the supremum is taken over \(P\in\mathbb{R}^{m\times r}\), \(Q\in\mathbb{R}^{n\times r}\), \(\mu,\nu\in\mathbb{R}^{r\times 1}\) satisfying \(\nu^T\mu=1\), such that \(\mathcal{W}(\rho)>0\) certifies the entanglement of \(\rho\). Second, our generalized block-matrix construction may be also used further to derive the measurement-based separability criteria, in which local measurement outcomes have been arranged into block matrices \cite{TangWu2023,LuPRA062424,LuAQT70344,LaiLuo2023,QiPangHou2025QIP}. It would also desirable to adopt our approach  to detect k-nonseparability and k-partite entanglement for multipartite quantum systems.

%

\appendix

\section{Proof of Lemma \ref{lem1}}\label{appendixA}

(1) The realignment map $\mathcal{R}$, the vectorization operator $\operatorname{vec}$, and the partial traces $\operatorname{Tr}_{A(B)}$ are all linear. From the definition of $\mathcal{Q}_{P,Q,\mu,\nu}$ in Eq. (\ref{realignment matrix}), we get
\begin{equation*}
\begin{aligned}
\mathcal{Q}_{P,Q,\mu,\nu}\left(\sum_{k=1} ^n p_k \rho_k\right)
&=\begin{pmatrix}
PQ^T & P\mu\operatorname{vec}\left(\operatorname{Tr}_A\left(\sum\limits_{k=1}^n p_k \rho_k\right)\right)^T \\
\operatorname{vec}\left(\operatorname{Tr}_B\left(\sum\limits_{k=1}^n p_k \rho_k\right)\right)\nu^TQ^T & \mathcal{R}\left(\sum\limits_{k=1}^n p_k \rho_k\right)
\end{pmatrix} \\
&=\begin{pmatrix}
\sum\limits_{k=1}^n p_k PQ^T & P\mu\operatorname{vec}\left(\sum\limits_{k=1}^n p_k \operatorname{Tr}_A(\rho_k)\right)^T \\
\operatorname{vec}\left(\sum\limits_{k=1}^n p_k \operatorname{Tr}_B(\rho_k)\right)\nu^TQ^T & \sum\limits_{k=1}^n p_k \mathcal{R}(\rho_k)
\end{pmatrix} \\
&=\begin{pmatrix}
\sum\limits_{k=1}^n p_k PQ^T & P\mu\sum\limits_{k=1}^n p_k \operatorname{vec}(\operatorname{Tr}_A(\rho_k))^T \\
\sum\limits_{k=1} ^n p_k \operatorname{vec}(\operatorname{Tr}_B(\rho_k))\nu^TQ^T & \sum\limits_{k=1}^n p_k \mathcal{R}(\rho_k)
\end{pmatrix} \\
&=\sum_{k=1}^n p_k
\begin{pmatrix}
PQ^T & P\mu\operatorname{vec}(\operatorname{Tr}_A(\rho_k))^T \\
\operatorname{vec}(\operatorname{Tr}_B(\rho_k))\nu^TQ^T & \mathcal{R}(\rho_k)
\end{pmatrix} \\
&=\sum_{k=1}^n p_k \mathcal{Q}_{P,Q,\mu,\nu}(\rho_k),
\end{aligned}
\end{equation*}
which proves (i) of Lemma \ref{lem1}.

(2) A general bipartite state $\rho$ can be written as $\rho=\sum\limits_k \xi_{k} \otimes \eta_k$, where $\xi_k \in \mathbb{C}^{d_A \times d_A}$ and $\eta_{k} \in \mathbb{C}^{d_B \times d_B}$. Denote $\sigma=(U \otimes V) \rho(U \otimes V)^{\dagger}$. We have
\begin{equation}\label{eq:A1}
\mathcal{Q}_{P,Q,\mu,\nu}(\sigma)=
\begin{pmatrix}
PQ^T & P\mu\operatorname{vec}(\operatorname{Tr}_A(\sigma))^T \\
\operatorname{vec}(\operatorname{Tr}_B(\sigma))\nu^TQ^T & \mathcal{R}(\sigma)
\end{pmatrix},
\end{equation}
where
\begin{align}\label{eq:A2}
P\mu\operatorname{vec}(\operatorname{Tr}_A(\sigma))^T&=P\mu\operatorname{vec}\left(\operatorname{Tr}_A\left((U \otimes V) \sum_k(\xi_k \otimes \eta_k)(U \otimes V)^{\dagger}\right)\right)^T \notag \\
&=P\mu\operatorname{vec}\left(\sum_k \operatorname{Tr}_A\left(U \xi_k U^{\dagger} \otimes V \eta_k V^{\dagger}\right)\right)^T \notag \\
&=P\mu\operatorname{vec}\left(\sum_k \operatorname{Tr}(\xi_k) V \eta_k V^{\dagger}\right)^T \notag \\
&=P\mu\sum_k \operatorname{Tr}(\xi_k) \operatorname{vec}\left(V \eta_k V^{\dagger}\right)^T \notag \\
&=P\mu \sum_k \operatorname{Tr}(\xi_k)\left((\overline{V} \otimes V) \operatorname{vec}(\eta_k)\right)^T \notag \\
&=P\mu \sum_k \operatorname{Tr}(\xi_k) \operatorname{vec}(\eta_k)^T(\overline{V} \otimes V)^T \notag \\
&=P\mu\operatorname{vec}\left(\sum_k \operatorname{Tr}_A(\xi_k \otimes \eta_k)\right)^T\left(\overline{V} \otimes V\right)^T \notag \\
&=P\mu \operatorname{vec}\left(\operatorname{Tr}_A(\rho)\right)^T\left(\overline{V} \otimes V\right)^T.
\end{align}
Here, we used $\operatorname{vec}(AXB)=\left(B^T\otimes A\right)\operatorname{vec}(X)$ for any $A\in\mathbb{C}^{m\times p}$, $X\in\mathbb{C}^{p\times q}$ and $B\in\mathbb{C}^{q\times n}$.

Similarly, we have
\begin{equation}\label{eq:A3}
\operatorname{vec}\left(\operatorname{Tr}_B(\sigma)\right)\nu^TQ^T=\left(\overline{U} \otimes U\right) \operatorname{vec}\left(\operatorname{Tr}_B(\rho)\right)\nu^TQ^T.
\end{equation}

Furthermore,
\begin{align}\label{eq:A4}
\mathcal{R}(\sigma)&=\sum_k\mathcal{R}\left((U \otimes V)\left(\xi_k \otimes \eta_k\right)\left(U \otimes V\right)^{\dagger}\right) \notag \\
&=\sum_k \mathcal{R}\left(U \xi_k U^{\dagger}\otimes V \eta_k V^{\dagger}\right)\notag\\
&=\sum_k \operatorname{vec}\left(U \xi_k U^{\dagger}\right) \operatorname{vec}\left(V \eta_k V^{\dagger}\right)^T \notag \\
&=\sum_k\left(\overline{U} \otimes U\right) \operatorname{vec}\left(\xi_k\right)\operatorname{vec}\left(\eta_k\right)^T\left(\overline{V} \otimes V\right)^T \notag \\
&=\sum_k(\overline{U} \otimes U) \mathcal{R}\left(\xi_k \otimes \eta_k\right)\left(\overline{V} \otimes V\right)^T \notag\\
&=\left(\overline{U} \otimes U\right) \mathcal{R}(\rho) \left(\overline{V} \otimes V\right)^T.
\end{align}
Here, we have used $\mathcal{R}(A\otimes B)=\operatorname{vec}(A)\operatorname{vec}(B)^T$ for matrices $A$, $B$.

Therefore, combining Eq. (\ref{eq:A1}), Eq. (\ref{eq:A2}), Eq. (\ref{eq:A3}) and Eq. (\ref{eq:A4}) we have
\begin{align*}
\mathcal{Q}_{P, Q, \mu, \nu}(\sigma)&=
\begin{pmatrix}
PQ^T & P\mu \operatorname{vec}\left(\operatorname{Tr}_A(\rho)\right)^T\left(\overline{V} \otimes V\right)^T \\
\left(\overline{U} \otimes U\right) \operatorname{vec}\left(\operatorname{Tr}_B(\rho)\right)\nu^TQ^T & \left(\overline{U} \otimes U\right) \mathcal{R}(\rho) \left(\overline{V} \otimes V\right)^T
\end{pmatrix} \\
&=\widetilde{U}\begin{pmatrix}
PQ^T & P\mu \operatorname{vec}\left(\operatorname{Tr}_A(\rho)\right)^T \\
\operatorname{vec}\left(\operatorname{Tr}_{B}(\rho)\right)\nu^TQ^T & \mathcal{R}(\rho)
\end{pmatrix}\widetilde{V} \\
&=\widetilde{U} \mathcal{Q}_{P,Q,\mu,\nu}(\rho) \widetilde{V},
\end{align*}
where
\begin{equation*}
\widetilde{U}=
\begin{pmatrix}
I_{m \times m} & 0 \\
0 & \overline{U} \otimes U
\end{pmatrix}, \qquad \widetilde{V}=
\begin{pmatrix}
I_{n \times n} & 0 \\
0 & V^{\dagger} \otimes V^T
\end{pmatrix}.
\end{equation*}

Hence,
\begin{equation*}
\left\|\mathcal{Q}_{P, Q, \mu, \nu}\left((U \otimes V) \rho\left(U \otimes V\right)^{\dagger}\right)\right\|_{\mathrm{tr}}=\left\|\widetilde{U} \mathcal{Q}_{P,Q,\mu, \nu}(\rho) \widetilde{V}\right\|_{\mathrm{tr}}=\left\|\mathcal{Q}_{P, Q, \mu, \nu}(\rho)\right\|_{\mathrm{tr}},
\end{equation*}
which proves (2) of Lemma \ref{lem1}.

\section{Proof of Lemma \ref{lem2}}\label{appendixB}

(1) Let $|\phi\rangle=\sum\limits_{k=0}^{d-1} \sqrt{\lambda_k}\left|a_k b_k\right\rangle=(U_A \otimes V_B)|\psi\rangle$, where $U_A$ and $V_B$ are unitary operators defined by $U_A|j\rangle = |a_k\rangle$ and $V_B|j\rangle = |b_k\rangle$ respectively, and $|\psi\rangle=\sum\limits_{k=0}^{d-1} \sqrt{\lambda_k}\left|k_A k_B\right\rangle$. As
\begin{equation*}
|\psi\rangle\langle\psi|=\sum\limits_{i=0}^{d-1} \sum\limits_{j=0}^{d-1} \sqrt{\lambda_i \lambda_j}\left|i_A i_B\right\rangle\left\langle j_A j_B\right|,
\end{equation*}
from (2) of Lemma \ref{lem1} one has
\begin{equation*}
\norm{\mathcal{Q}_{P, Q, \mu, \nu}(|\phi\rangle\langle\phi|)}_{\mathrm{tr}}=\norm{\mathcal{Q}_{P, Q, \mu, \nu}(|\psi\rangle\langle\psi|)}_{\mathrm{tr}}=
\begin{pmatrix}
PQ^T & P\mu\operatorname{vec}\left(\operatorname{Tr}_A(|\psi\rangle\langle\psi|)\right)^T \\
\operatorname{vec}\left(\operatorname{Tr}_B(|\psi\rangle\langle\psi|)\right)\nu^TQ^T & \mathcal{R}(|\psi\rangle\langle\psi|)
\end{pmatrix}.
\end{equation*}

Let
\begin{equation*}
\mathrm{Q}_1=
\begin{pmatrix}
PQ^T & P\mu M \\
M^T\nu^TQ^T & \Omega_1
\end{pmatrix},~~
\mathrm{Q}_2=
\begin{pmatrix}
0 & 0 \\
0 & \Omega_2
\end{pmatrix},
\end{equation*}
where
\begin{equation*}
M=(\lambda_0,\underbrace{0,\cdots,0}_{d-1},0,\lambda_1,\underbrace{0,\cdots,0}_{d-2}, 0, 0, \lambda_2, \underbrace{0, \cdots, 0}_{d-3}, \cdots,\underbrace{0,\cdots,0}_{d-1},\lambda_{d-1})_{1 \times d^2},
\end{equation*}

\begin{equation*}
\Omega_1=\operatorname{diag}(\lambda_0, \underbrace{0, \cdots, 0}_{d-1}, 0, \lambda_{1}, \underbrace{0, \cdots, 0}_{d-2}, 0, 0, \lambda_2, \underbrace{0, \cdots, 0}_{d-3},\cdots, \underbrace{0, \cdots, 0}_{d-1}, \lambda_{d-1})\in \mathbb{C}^{d^2 \times d^2},
\end{equation*}

\begin{equation*}
\begin{aligned}
\Omega_{2}=&\operatorname{diag}\left(0, \sqrt{\lambda_{0} \lambda_{1}}, \cdots, \sqrt{\lambda_{0} \lambda_{d-1}}, \sqrt{\lambda_{1} \lambda_{0}}, 0, \sqrt{\lambda_{1} \lambda_{2}}, \cdots, \sqrt{\lambda_{1} \lambda_{d-1}}, \right. \\
& \left.\sqrt{\lambda_{2} \lambda_{0}}, \sqrt{\lambda_{2} \lambda_{1}}, 0, \sqrt{\lambda_{2} \lambda_{3}}, \cdots, \sqrt{\lambda_{2} \lambda_{d-1}}, \cdots, \sqrt{\lambda_{d-1} \lambda_{0}}, \cdots, \sqrt{\lambda_{d-1} \lambda_{d-2}}, 0\right)\in \mathbb{C}^{d^2 \times d^2}.
\end{aligned}
\end{equation*}

Using the definition of trace norm we get
\begin{equation*}
\left\|\mathcal{Q}_{P,Q,\mu, \nu}(|\phi\rangle\langle\phi|)\right\|_{\mathrm{tr}}=\left\|\mathrm{Q}_1\right\|_{\mathrm{tr}}+\left\|\mathrm{Q}_2\right\|_{\mathrm{tr}},
\end{equation*}
where $\left\|\mathrm{Q}_2\right\|_{\mathrm{tr}}=2\sum\limits_{0 \leq i<j \leq d-1} \sqrt{\lambda_i \lambda_j}$. Since $\Omega_1=\mathcal{R}\left(\Omega_1\right)$ and $\Omega_1$ is a separable state, according to Theorem \ref{thm1} we get
\begin{equation*}
\left\|\mathrm{Q}_1\right\|_{\mathrm{tr}}=\left\|\mathcal{Q}_{P,Q,\mu, \nu}\left(\Omega_1\right)\right\|_{\mathrm{tr}} \leq \mathfrak{M}({P,Q,\mu,\nu}).
\end{equation*}
Therefore,
\begin{equation*}
\left\|\mathcal{Q}_{P, Q, \mu, \nu}(|\phi\rangle\langle\phi|)\right\|_{\mathrm{tr}}\leq\mathfrak{M}({P,Q,\mu,\nu})+2\sum_{0 \leq i<j \leq d-1} \sqrt{\lambda_i \lambda_j}.
\end{equation*}

(2) Note that for any $\lambda_k \geq 0$, $k=0, 1, \cdots, d-1$, such that $\lambda_0 + \lambda_1 + \cdots + \lambda_{d-1} = 1$, one has
\begin{equation*}
2\sum_{0 \leq i<j \leq d-1} \sqrt{\lambda_i \lambda_j} =\left(\sum_{k=0}^{d-1}\sqrt{\lambda_k}\right)^2-1 \leq d-1.
\end{equation*}
Based on the relation above and (1) of Lemma \ref{lem2}, we get
\begin{equation*}
\left\|\mathcal{Q}_{P, Q, \mu, \nu}(|\phi\rangle\langle\phi|)\right\|_{\mathrm{tr}}\leq\mathfrak{M}({P,Q,\mu,\nu})+2 \sum_{0 \leq i<j \leq d-1} \sqrt{\lambda_i\lambda_j}\leq\mathfrak{M}({P,Q,\mu,\nu})+(d-1).
\end{equation*}
Therefore, we complete the proof of (2) in Lemma \ref{lem2}.

\section{Proof of Theorem \ref{thm6}}\label{appendixC}

\begin{proof}
For a pure state $|\phi\rangle$ in $\mathbb{C}^{d} \otimes \mathbb{C}^{d} \otimes \mathbb{C}^{d}$, using Eq.~(\ref{concurrence}) and Theorem \ref{thm2} we get
\begin{equation}\label{C:1}
\sqrt{1-\operatorname{Tr}\left(\rho_{1}^{2}\right)} \geq \frac{1}{\sqrt{(d-1)d\,}}\Big(\left\|\mathcal{Q}_{P,Q,\mu,\nu}^{1|23}\left(|\phi\rangle\langle\phi|\right)\right\|_{\mathrm{tr}}-\mathfrak{M}({P,Q,\mu,\nu})\Big),
\end{equation}
\begin{equation}\label{C:2}
\sqrt{1-\operatorname{Tr}\left(\rho_{2}^{2}\right)} \geq \frac{1}{\sqrt{(d-1)d\,}}\Big(\left\|\mathcal{Q}_{P,Q,\mu,\nu}^{2|13}
\left(|\phi\rangle\langle\phi|\right)\right\|_{\mathrm{tr}}-\mathfrak{M}({P,Q,\mu,\nu})\Big)
\end{equation}
and
\begin{equation}\label{C:3}
\sqrt{1-\operatorname{Tr}\left(\rho_{3}^{2}\right)} \geq \frac{1}{\sqrt{(d-1)d\,}}\Big(\left\|\mathcal{Q}_{P,Q,\mu,\nu}^{3|12}
\left(|\phi\rangle\langle\phi|\right)\right\|_{\mathrm{tr}}-\mathfrak{M}({P,Q,\mu,\nu})\Big).
\end{equation}

Combining (\ref{C:1}), (\ref{C:2}) and (\ref{C:3}) we get
\begin{equation*}
\begin{aligned}
&\qquad3\sqrt{(d-1)d}\sqrt{1-\operatorname{Tr}\left(\rho_{1}^{2}\right)}-3\mathscr{Q}_{P,Q,\mu,\nu}(|\phi\rangle\langle\phi|)+3\,\mathfrak{M}(P,Q,\mu,\nu)+2(d-1) \\
&=3\sqrt{d(d-1)} \sqrt{1-\operatorname{Tr}\left(\rho_{1}^{2}\right)}-\Big(\left\|\mathcal{Q}_{P,Q,\mu,\nu}^{1|23}\left(|\phi\rangle\langle\phi|\right)\right\|_{\mathrm{tr}}+\left\|\mathcal{Q}_{2|13}\left(|\phi\rangle\langle\phi|\right)\right\|_{\mathrm{tr}}+\left\|\mathcal{Q}_{3|12}\left(|\phi\rangle\langle\phi|\right)\right\|_{\mathrm{tr}}\Big) \\
&\quad+3\,\mathfrak{M}(P,Q,\mu,\nu)+2(d-1) \\
&\geq 3\sqrt{(d-1)d} \sqrt{1-\operatorname{Tr}\left(\rho_{1}^{2}\right)}-\sqrt{(d-1)d\,}\left(\sqrt{1-\operatorname{Tr}\left(\rho_{1}^{2}\right)}+\sqrt{1-\operatorname{Tr}\left(\rho_{2}^{2}\right)}+\sqrt{1-\operatorname{Tr}\left(\rho_{3}^{2}\right)}\,\right)+2(d-1) \\
&=2\sqrt{d(d-1)} \sqrt{1-\operatorname{Tr}\left(\rho_{1}^{2}\right)}-\sqrt{d(d-1)}
\left(\sqrt{1-\operatorname{Tr}\left(\rho_{2}^{2}\right)}+\sqrt{1-\operatorname{Tr}\left(\rho_{3}^{2}\right)}\,\right)+2(d-1).
\end{aligned}
\end{equation*}

Since $\sqrt{1-\operatorname{Tr}\left(\rho_{1}^{2}\right)} \geq 0$ and $\sqrt{1-\operatorname{Tr}\left(\rho_{k}^{2}\right)} \leq \sqrt{1-\frac{1}{d}}$ $(k=2,3)$, we obtain
\begin{equation*}
\begin{aligned}
& \quad 3\sqrt{(d-1)d\,}\sqrt{1-\operatorname{Tr}\left(\rho_{1}^{2}\right)}-3\mathscr{Q}_{P,Q,\mu,\nu}(|\phi\rangle\langle\phi|)+3\,\mathfrak{M}(P,Q,\mu,\nu)+2(d-1) \\
&\geq 2\sqrt{(d-1)d\,} \sqrt{1-\operatorname{Tr}\left(\rho_{1}^{2}\right)}-\sqrt{d(d-1)}
\left(\sqrt{1-\operatorname{Tr}\left(\rho_{2}^{2}\right)}+\sqrt{1-\operatorname{Tr}\left(\rho_{3}^{2}\right)}\,\right)+2(d-1) \\
& \geq 2(d-1)-2\sqrt{d(d-1)}\sqrt{1-\frac{1}{d}} \\
& = 0.
\end{aligned}
\end{equation*}
Therefore
\begin{equation*}
\sqrt{1-\operatorname{Tr}\left(\rho_{1}^{2}\right)} \geq \frac{1}{\sqrt{(d-1)d\,}}\bigg(\mathscr{Q}_{P,Q,\mu,\nu}(|\phi\rangle\langle\phi|)-\mathfrak{M}(P,Q,\mu,\nu)-\frac{2(d-1)}{3}\bigg).
\end{equation*}

Similarly, for $k=2,3$, we have
\begin{equation*}
\sqrt{1-\operatorname{Tr}\left(\rho_{k}^{2}\right)} \geq \frac{1}{\sqrt{(d-1)d\,}}\bigg(\mathscr{Q}_{P,Q,\mu,\nu}(|\phi\rangle\langle\phi|)-\mathfrak{M}(P,Q,\mu,\nu)-\frac{2(d-1)}{3}\bigg).
\end{equation*}
	
Now let $\left\{p_{k},\left|\phi_{k}\right\rangle\right\}$ be the optimal decomposition for $\rho$ such that $C_{\text{GME}}(\rho)=\sum\limits_{k} p_{k} C_{\text{GME}}\left(\left|\phi_{k}\right\rangle\right)$. Then
\begin{equation*}
\begin{split}
C_{\text{GME}}(\rho)&=\sum\limits_{k} p_{k} C_{\text{GME}}\left(\left|\phi_{k}\right\rangle\right) \\
& \geq \frac{1}{\sqrt{(d-1)d\,}} \sum\limits_{k} p_{k}\bigg(\mathscr{Q}_{P,Q,\mu,\nu}(|\phi_{k}\rangle\langle\phi_{k}|)-\mathfrak{M}(P,Q,\mu,\nu)-\frac{2(d-1)}{3}\bigg) \\
& \geq \frac{1}{\sqrt{(d-1)d\,}}\bigg(\mathscr{Q}_{P,Q,\mu,\nu}(\rho)-\mathfrak{M}(P,Q,\mu,\nu)-\frac{2(d-1)}{3}\bigg),
\end{split}
\end{equation*}
which completes the proof of Theorem \ref{thm6}.
\end{proof}

\end{document}